\documentclass{techrep}

\usepackage{amsfonts}
\usepackage{amsthm}

\newtheorem{theorem}{Theorem}{\bfseries}{\itshape}
\newtheorem{lemma}{Lemma}{\bfseries}{\itshape}
{\bfseries}{\itshape}
{\bfseries}{\itshape}
\newtheorem{obs}{Observation}{\bfseries}{\itshape}

\theoremstyle{definition}

\theoremstyle{plain}

\usepackage{verbatim}
\usepackage[utf8]{inputenc}
\usepackage{microtype}
\usepackage{graphicx}
\usepackage{caption}
\usepackage{subcaption}
\usepackage{hyperref}
\usepackage[table,dvipsnames]{xcolor}
\definecolor{linkblue}{named}{Blue}
\hypersetup{colorlinks=true, linkcolor=linkblue,  anchorcolor=linkblue,
	citecolor=linkblue, filecolor=linkblue, menucolor=linkblue,
	urlcolor=linkblue} 
\usepackage{wasysym}
\usepackage{graphicx}
\usepackage{caption}
\usepackage{enumitem}
\usepackage{thmtools, thm-restate}
\usepackage{wrapfig}
\usepackage{float}
\usepackage{stackrel}
\usepackage{fixfoot}
\usepackage[english]{babel}
\usepackage{mathrsfs}  
\usepackage{venturis}
\usepackage{placeins}
\usepackage{algpseudocode}
\usepackage{algorithm}
\usepackage{mathtools}
\usepackage[mathlines]{lineno}

\newcommand{\R}{\mathbb{R}}

\DeclarePairedDelimiter{\abs}{\lvert}{\rvert}

\DeclarePairedDelimiter{\floor}{\lfloor}{\rfloor}
\DeclarePairedDelimiter{\ceil}{\lceil}{\rceil}

\let\epsilon\varepsilon

\newcommand{\jit}[1]{}

\newcommand{\vomit}[1]{}

\begin{document}

\title{Computational aspects of disks enclosing many points\footnote{This is an extended version of a paper appearing in Proceedings of LAGOS 2025.}}
\author{%
  Prosenjit Bose%
  \thanks{\affil{Carleton University}, 
          \email{jit@scs.carleton.ca},
          \email{tylertuttle@cmail.carleton.ca}},
  Guillermo Esteban%
  \thanks{\affil{Universidad de Alcalá},
          \email{g.esteban@uah.es}},
  and Tyler Tuttle\footnotemark[2]
}
\date{}
\maketitle

\begin{abstract}
Let $S$ be a set of $n$ points in the plane. We present several different algorithms for finding a pair of points in $S$ such that any disk that contains that pair must contain at least $cn$ points of $S$, for some constant $c>0$. The first is a randomized algorithm that finds a pair in $O(n\log n)$ expected time for points in general position, and $c = 1/2-\sqrt{(1+2\alpha)/12}$, for any $0<\alpha<1$. The second algorithm, also for points in general position, takes quadratic time, but the constant $c$ is improved to $1/2-1/{\sqrt{12}} \approx 1/4.7$. The second algorithm can also be used as a subroutine to find the pair that maximizes the number of points inside any disk that contains the pair, in $O(n^2\log n)$ time. We also consider variants of the problem. When the set $S$ is in convex position, we present an algorithm that finds in linear time a pair of points such that any disk through them contains at least $n/3$ points of $ S $. For the variant where we are only interested in finding a pair such that the diametral disk of that pair contains many points, we also have a linear-time algorithm that finds a disk with at least $n/3$ points of $S$. Finally, we present a generalization of the first two algorithms to the case where the set $S$ of points is coloured using two colours. We also consider adapting these algorithms to solve the same problems when $S$ is a set of points inside of a simple polygon $P$, with the notion of a disk replaced by that of a geodesic disk.
\end{abstract}

\section{Introduction}
\label{sec:introduction}

Let $S$ be a set of $n$ points in the plane. A disk \emph{contains} or \emph{encloses} a point $p \in S $ if $p$ is on the boundary or in the interior of the disk. For any two distinct points $p, q \in S$, we define $C_S(p,q)$ to be the smallest integer such that any disk enclosing $p$ and $q$ encloses at least $C_S(p,q)$ points of~$S$. Let $\Pi(S)$ be the largest value of $C_S(p,q)$ over all pairs $\{p,q\}$ of points in $S$, and let $\Pi(n)$ be the minimum value of $\Pi(S)$ over all sets~$S$ of $n$ points. The problem of finding the exact value of~$\Pi(n)$ was introduced in 1988 by Neumann-Lara and Urrutia~\cite{neumann-lara88}. They show that there exists a pair of points in $S$ such that any disk that contains that pair must contain at least $cn$ points of $S$, for some constant $c>1/60$. Over the next year or so, the constant factor in this result was improved, culminating in a proof by Edelsbrunner et al.~\cite{edelsbrunner89} that there exists a pair of points such that any disk containing that pair also contains at least $n/4.7$ points of $S$. Recently, Bose et al.\ studied this problem in the setting of a simple polygon instead of in the plane~\cite{wadspaper}.


These results all prove the existence of pairs of points, but none of them consider the algorithmic problem of finding such a pair. That is the focus of this paper. Many of the ideas used to prove existence will be useful in providing us with algorithms both in the Euclidean setting and the geodesic setting.

\subsection{Previous work}
\label{sec:previous-work}  
The combinatorial problem was introduced by Neumann-Lara and Urrutia~\cite{neumann-lara88} in 1988, who proved that there exists a pair of points $p$, $q$ in $S$ with the property that any disk that contains $p$ and $q$ contains at least $\ceil{(n-2)/60}$ points of $S$. They also defined $\overline{\Pi}(n)$ to be the minimum value of $\Pi(S)$ over all sets $S$ of $n$ points in \textit{convex position}, showing that $\overline{\Pi}(n) \ge \ceil{(n-2)/4}$. Subsequently, B\'ar\'any et al.~\cite{barany89} 
improved the lower bound to $\Pi(n) \ge n/30$.
Simultaneously, the problem was considered by Hayward et al.~\cite{hayward89a}, who proved that $\floor{n/27} + 2 \le \Pi(n) \le \ceil{n/4} + 1$, and $\overline{\Pi}(n) = \ceil{n/3} + 1$. The lower bound for point sets in non-convex position was later improved by Hayward~\cite{hayward89b} to $\Pi(n) \ge \ceil{5(n - 2)/84}$.
Finally, Edelsbrunner et al.~\cite{edelsbrunner89} proved that
\begin{equation*}
\Pi(n) \ge \frac{n - 3}{2} - \sqrt\frac{(n-2)^2 - 1}{12} = \bigg(\frac{1}{2} - \frac{1}{\sqrt{12}}\bigg)n + O(1) \approx \frac{n}{4.7},
\end{equation*}
which stands as the best lower bound on $\Pi(n)$. An alternate proof of the same bound was given by Ramos and Via\~na~\cite{ramos09}, who showed the stronger result that there always exist two points $p$ and $q$ with the property that any disk which has $p$ and $q$ on its boundary contains at least $n/4.7$ points, and there are at least $n/4.7$ points outside of the disk. Another proof of this fact was given by Claverol et al.~\cite{claverol21} using higher-order Voronoi diagrams.

Variations of the problem have been studied over the years. Akiyama et al.~\cite{akiyama96} proved that there exist two points $p$ and $q$ such that the disk with segment $pq$ as a diameter encloses at least $n/3$ points. B\'ar\'any et al.~\cite{barany89} considered the problem in $\R^d$. Specifically, they showed that given a set~$S$ of $n$ points in $\R^d$, there exists a subset $T$ of $S$ of size $\floor{(d+3)/2}$ such that any d-dimensional ball that contains $T$ contains at least $cn$ points of $S$, for some constant $c > 0$. B{\'a}r{\'a}ny and Larman generalized this further from balls to quadric surfaces~\cite{barany1990combinatorial}. In addition, a bichromatic version of the problem has been studied in~\cite{claverol21,prodromou2007combinatorial,urrutiaproblemas}, where the goal is to find a bichromatic pair of points such that any disk that contains the pair must contain many points. The author of~\cite{prodromou2007combinatorial} also generalizes this coloured version of the problem to higher dimensions.

Many of these variations were generalized from the plane to sets of points contained in a simple polygon by Bose et al.\ \cite{wadspaper}. They show that given a set $S$ of $n$ points in a simple polygon~$P$, there exist two points of~$S$ such that any geodesic disk that contains those two points on its boundary contains at least $n/5$ points of~$S$. For point sets in geodesically convex position, they prove that there exists a pair such that any geodesic disk containing the pair contains at least $n/3$ points, matching the bound of the Euclidean setting. They also consider the diametral and bichromatic versions of the problem, as well as the version for points both inside and outside the disk.

\subsection{Our results}

In this paper, we present different algorithms which, given a set $S$ of $n$ points, finds a pair of points such that any disk containing those points contains a constant fraction of the points of $S$. The first two algorithms will be for sets of points in general position. One is a randomized algorithm which finds, in $O(n\log n)$ expected time, a pair of points $p$ and $q$ such that any disk that contains~$p$ and $q$ contains at least $cn$ points of $S$, for any constant $0<c<1/4.7$. The constant in the running time of this algorithm depends on the value of $c$. The second is a quadratic-time algorithm that finds a pair of points $p$ and $q$ such that any disk that contains $p$ and $q$ contains at least $n/4.7$ points of $S$, matching the best-known lower bound for $\Pi(n)$ This second algorithm can also be used to find the pair that maximizes $C_S(p,q)$ in $O(n^2\log n)$ time. The third algorithm is for points in convex position, and finds a pair of points $p$ and $q$ such that any disk that contains $p$ and $q$ contains at least $n/3$ points of $S$, which is the optimal bound for $\overline{\Pi}(n)$. The third algorithm runs in linear time. Fourth we have an algorithm for the diametral version of the problem, which finds two $p$ and $q$ such that the disk with $pq$ as a diameter contains at least $n/3$ points of $S$ in $O(n)$ time. Finally, we show how to modify the first two algorithms to work for the bichromatic version of the problem: given a set $S$ of $n/2$ red and $n/2$ blue points, find a red-blue pair such that any disk containing that pair contains at least $cn$ points of $S$, for a constant $c$.

We also discuss how these algorithms can be generalized for sets of points in a simple polygon. Care must be taken since many properties that are trivial in the plane become more difficult in this setting. For example, finding the centre of the disk through three points in the plane takes constant time. In a simple polygon such a disk may not even exist. We will show how to handle all of the difficulties that arise.

\section{Background}
\label{sec:background}

In this section, we give the required background needed for our algorithms, both in the plane (Euclidean setting) and in a simple polygon (geodesic setting).

\subsection{Euclidean setting}

Given two points $p$ and $q$, we denote the line segment joining $p$ and $q$ by $pq$, its length by $\lvert pq\rvert$, the line through $p$ and $q$ by $L(p,q)$, and the bisector of $p$ and $q$ by $b(p,q)$. The line $L(p,q)$ splits the plane into two open halfplanes. Let $H(p,q)$ be the halfplane to the left of $L(p,q)$ oriented from~$p$ to $q$, and $H(q,p)$ be the halfplane on the right side of $L(p,q)$. We say that a disk $D$ \emph{contains} a point $p$ if $p$ lies in the interior or on the boundary of $D$. We say that $D$ \emph{passes through} $p$ if $p$ lies on the boundary of~$D$.

Given a set $S$ of points in the plane and two points $p$ and $q$ of $S$, the minimum integer $k$ such that any disk that contains $p$ and $q$ contains at least $k$ other points of $S$ is denoted $C_S(p,q)$. The minimum integer $k$ such that any disk with $p$ and $q$ on its boundary contains at least $k$ other points of $S$ and at most $n-k-2$ other points of $S$ (or equivalently at least $k$ points of $S$ are outside the disk) is denoted $\tilde{C}_S(p,q)$.

For the rest of the paper we will assume that~$S$ is in general position, meaning no three points of $S$ are collinear and no four points of $S$ are cocircular.

\subsection{Geodesic setting}

Let $P$ be a simple polygon. Given a pair of points $p$ and $q$ in $P$, the \emph{geodesic (shortest) path} from $p$ to $q$ is denoted $g(p,q)$. It is known that $g(p,q)$ is a polygonal chain~\cite{DBLP:journals/networks/LeeP84}. The length of this path is denoted $\abs{g(p,q)}$. The set of all points in $P$ equidistant from both $p$ and $q$ is called the \emph{bisector} of $p$ and $q$, which we denote $b(p,q)$. Under the assumption that no vertex of $P$ is equidistant from $p$ and $q$, it is known that $b(p,q)$ is a simple piecewise-defined curve made up of line segments and hyperbolic arcs~\cite{DBLP:journals/algorithmica/Aronov89}. The \emph{extension path} from $p$ to $q$, denoted $\ell(p,q)$, is defined by extending the first and last segments of $g(p,q)$ until they properly intersect the boundary of $P$. The extension path $\ell(p,q)$ naturally partitions the polygon into three parts: the set of points to the left of $\ell(p,q)$, the set of points to the right of $\ell(p,q)$, and $\ell(p,q)$ itself. We say that a point $u$ is to the left (resp.\ right) of $g(p,q)$ if $u$ is to the left (resp.\ right) of $\ell(p,q)$.

A subset $Q$ of $P$ is called \emph{geodesically convex} if, for all pairs of points $p$ and $q$ in $Q$, the geodesic path $g(p,q)$ is contained in $Q$. The \emph{geodesic convex hull} of a set $S$ of points in~$P$ is the smallest geodesically convex set that contains $S$. We say that $S$ is in \emph{geodesically convex position} if every point of $S$ is on the boundary of the geodesic convex hull of $S$. The \emph{geodesic triangle} on three points $ p,q,u \in P $, denoted $\triangle pqu$, is the region bounded by $ g(p,q) $, $ g(q,u) $, and $ g(u,p) $.

A \emph{geodesic disk} centred at $ o \in P $ with radius $ \rho \geq 0 $ is the set $ D(o,\rho) = \{p \in P : |g(o,p)| \leq \rho \} $. Similarly to in the plane, we say that $D(o,\rho)$ contains $p$ if $\abs{g(o,p)}\le \rho$, and that $D(o,\rho)$ passes through $p$ if $\abs{g(o,p)}=\rho$.

A set $S$ of at least three points is called \emph{geodesically collinear} if there exist two points $p$ and $q$ in $S$ such that $S$ is a subset of $g(p,q)$. A set $S$ of at least four points is called \emph{geodesically cocircular} if there exists a geodesic disk $D(o,\rho)$ such that $S$ is contained on its boundary, i.e.\ $\abs{g(o,p)}=\rho$ for all $p$ in $S$. As in the planar case, we make a general position assumption about any set $S$ of points in a simple polygon: no three points are geodesically collinear, no four points are geodesically cocircular, and no two points are equidistant from a vertex of $P$.

\subsection{Higher-order Voronoi diagrams}

Since many of our results in both the Euclidean and geodesic settings rely on higher-order Voronoi diagrams, we give a brief overview. For a more detailed treatment, see Edelsbrunner~\cite{DBLP:series/eatcs/Edelsbrunner87}. We will define the diagrams in the Euclidean setting, but note that the definitions all carry over to the geodesic setting with appropriate modifications: Euclidean distance replaced by geodesic distance, etc.

Let $S$ be a set of points in the plane. The Voronoi region of a point $p$ of $S$ is the set of all points in the plane closer to $p$ than to any other point of $S$. It is well-known that the Voronoi region of $p$ is nonempty (since it must contain $p$) and connected, forms the interior of a polygon, and the edges of this polygon are segments of bisectors of the form $b(p,q)$, where $q$ is another point of $S$. The Voronoi diagram of $S$ is the planar subdivision defined by these Voronoi regions, and we denote it by $V(S)$.

The order-$k$ Voronoi diagram is defined as follows. Let $T$ be a subset of $S$ of size $k$. The order-$k$ Voronoi region of $T$ is the set of all points in the plane closer to every point of $T$ than to any point of $S\setminus T$. That is, the set of points $r$ such that $\abs{pr} < \abs{qr}$ for all points $p$ in $T$ and $q$ in $S\setminus T$. Note that the standard Voronoi diagram is the order-$1$ Voronoi diagram. We denote the order-$k$ Voronoi diagram by $V_k(S)$. The order-$(n-1)$ Voronoi diagram is also called the farthest-point Voronoi diagram.

Unlike the order-$1$ Voronoi diagram, the Voronoi region of a set $T$ may be empty. However, like the order-$1$ Voronoi regions any nonempty order-$k$ Voronoi region must be the interior of a polygon whose edges are bisectors of the form $b(p,q)$, where $p$ is in $T$ and $q$ is in $S\setminus T$.

\section{Algorithms in the plane}

In this section we present three algorithms for finding a pair of points $ p $ and $q$ in the Euclidean plane with $C_S(p,q)\ge cn$, for some $c > 0$.

\subsection{Computing the minimum weight segment on a bisector}
\label{sec:weightplane}

We first present an algorithm that, given a set $S$ of $n$ points and two points $p$ and $q$ in $S$, computes $C_S(p,q)$ and $\tilde{C}_S(p,q)$ in $O(n\log n)$ time.
Consider the set of all disks through $p$ and $q$. There are $n-2$ such disks that pass through a third point of $S$.
The centres of these disks all lie on $b(p,q)$, and split the line $b(p,q)$ into $n-1$ open segments, two of which are unbounded. Two disks through $p$ and $q$ with centre on the same segment~$s$ contain the same set of points. The \emph{weight} of $s$ is defined to be the number of points in the interior of any disk through $p$ and $q$ centred on $s$. We denote the weight of $s$ by $\omega(s)$. Note that $p$ and $q$ are on the boundary of the disk, not in the interior, so if no other point is inside the disk then $\omega(s)$ is equal to $0$. Also note that a circle through $p$ and $q$ whose centre is on $s$ will have $n-\omega(s)-2$ points on its exterior. Clearly, we have
\begin{align}
    C_S(p,q) &= \min\{\omega(s) : \text{$s$ is a segment of $b(p,q)$}\}\text{, and} \\
    \tilde{C}_S(p,q) &= \min\{\min\{\omega(s), n-\omega(s)-2\} : \text{$s$ is a segment of $b(p,q)$}\}.
\end{align}
Furthermore, two segments $s$ and $s'$ that share an endpoint differ in weight by exactly one. This last observation is key to our algorithm: we can sort the segments along $b(p,q)$, compute the weight of one of the unbounded segments, and then it computes the weights of the segments in order along the bisector efficiently.

\begin{algorithm}
\caption{Computing $C_S(p,q)$ and $\tilde{C}_S(p,q)$}
\label{alg:one}
\begin{algorithmic}[1]
\State For each point $x_i$ in $S\setminus \{p,q\}$, compute the centre $c_i$ of the circle through $p$, $q$, and $x_i$
\State Sort $c_1, c_2, \dots, c_{n-2}$ along $b(p,q)$ from left to right, oriented according to $L(p,q)$
\State $\omega \gets \abs{S \cap H(p,q)}$
\State $C \gets \omega$
\State $\tilde{C} \gets \min\{\omega, n - \omega - 2\}$
\For{$i \gets 1, 2, \dots, n-2$}
    \If{$x_i$ is to the left of $L(p,q)$}
        \State $\omega \gets \omega - 1$
    \Else
        \State $\omega \gets \omega + 1$
    \EndIf
    \State $C \gets \min\{C, \omega\}$
    \State $\tilde{C} \gets \min\{\tilde{C}, \omega, n - \omega - 2\}$
\EndFor
\State \textbf{return} $C$ and $\tilde{C}$
\end{algorithmic}
\end{algorithm}

\begin{figure}
    \centering
    \includegraphics{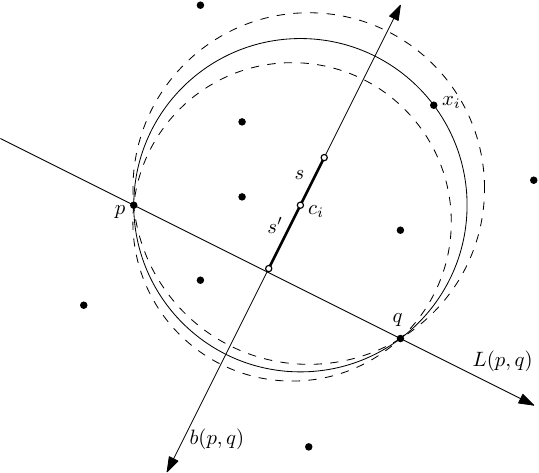}
    \caption{$x_i$ is to the left of $L(p,q)$, thus $\omega(s) = 5$ and $\omega(s') = 4$. The dashed disks are disks through $ p $ and $ q $ and centred at a point in $ s $ and $ s' $.}
    \label{fig:segment-weights}
\end{figure}

\begin{lemma}
\label{lem:alg1}
Let $S$ be a set of $n$ points in the plane, and let $p$ and $q$ be points of $S$. The values of $C_S(p,q)$ and $\tilde{C}_S(p,q)$ can be computed in $O(n \log n)$ time.
\end{lemma}

\begin{proof}
See Algorithm~\ref{alg:one}. Consider the unbounded segment $s$ of $b(p,q)$ bounded by $c_1$. A disk through $p$ and $q$ whose centre is on this segment contains all the points in $S \cap H(p,q)$. Therefore, $\omega(s) = \abs{S \cap H(p,q)}$.

Suppose we know the weight of a segment $s$, and let $s'$ be the adjacent segment along $b(p,q)$, such that~$s'$ is to the right of $s$ from the perspective of the $L(p,q)$. Let $x_i$ be the point of $S$ such that $c_i$ is the shared endpoint of $s$ and $s'$, see Figure~\ref{fig:segment-weights}. If $x_i$ is to the left of the line $L(p,q) $ oriented from $ p $ to $ q $, then $\omega(s')=\omega(s)-1$. Otherwise, if $u$ is to the right of $L(p,q)$, then $\omega(s')=\omega(s)+1$. Thus we can compute the value of $\omega(s')$ from the value of $\omega(s)$ in constant time. Clearly then the algorithm will correctly compute the values of $C_S(p,q)$ and $\tilde{C}_S(p,q)$.

Constructing the list $\{c_1, c_2, \dots, c_{n-2}\}$ of centres takes $O(n)$ time, since finding the centre of a circle through three points in the plane can be done in constant time. Sorting these centres along $b(p,q)$ takes $O(n \log n)$ time. Computing the initial value for $\omega$ of $\abs{S\cap H(p,q)}$ takes $O(n)$ time. Finally, each iteration of the loop takes constant time so the entire loop takes $O(n)$ time. The total running time of this algorithm is dominated by the time it takes to sort the centres.
\end{proof}

\subsection{A randomized algorithm for finding a pair of points}

Now that we have an algorithm for computing $C_S(p,q)$ and $\tilde{C}_S(p,q)$, we present an algorithm for finding a pair $\{p,q\}$ with $\tilde{C}_S(p,q) \ge cn$, for some constant $c > 0$: repeatedly choose a pair $\{p,q\}$ uniformly at random and compute $\tilde{C}_S(p,q)$ until the value of $\tilde{C}_S(p,q)$ is at least $cn$. Notice that this also gives an algorithm for finding a pair with $C_S(p,q)\ge cn$, since $C_S(p,q)\ge\tilde{C}_S(p,q)$. Since---at least in the plane---the best-known lower bounds on $\Pi(n)$ and $\tilde{\Pi}(n)$ are equal up to lower order terms, we only present an algorithm for the case of points inside and outside any disk through $p$ and $q$.

Consider the set of $\binom{n}{2}$ bisectors of pairs of points of $S$. Let $S_k$ be the number of bisectors that have at least one segment of weight less than $k$ or more than $n-k-2$. Ramos and Via\~na~\cite{ramos09} proved that $S_k \le 3(k+1)n-3(k+1)(k+2)$.

\begin{theorem}
\label{thm:alg1}
Let $0<\alpha<1$ be a constant, and let $S$ be a set of $n$ points in the plane.
There are at least $\alpha\binom{n}{2}$ pairs $\{p,q\}$ of points of $S$ such that $\tilde{C}_S(p,q)\ge (1/2-\sqrt{(1+2\alpha)/12})n + O(1)$. Such a pair can be found in $O(n \log n)$ expected time, or with probability $1-1/n$ in $O(n\log^2n)$ time.
\end{theorem}

\begin{proof}
If $S_k < (1-\alpha)\binom{n}{2}$, then there must be at least $\alpha\binom{n}{2}$ bisectors with no segments of weight less than $k$ or greater than $n-k-2$, meaning there are that many pairs $\{p,q\}$ such that every disk through $p$ and $q$ contains at least $k$ points and at most $n-k-2$ points.

Using Ramos and Via\~na's upper bound on $S_k$ we note that $S_k<(1-\alpha)\binom{n}{2}$ as long as $3(k+1)n-3(k+1)(k+2) < (1-\alpha)\binom{n}{2}$. This is a quadratic equation in $k$ whose smallest solution is
\begin{equation}
    k < \frac{n-3}{2} - \bigg(\frac{(1+2\alpha)n^2 - (4+2\alpha)n + 15}{12}\bigg)^{1/2} = \bigg(\frac{1}{2}-\sqrt{\frac{1+2\alpha}{12}}\bigg)n + O(1).
\end{equation}
The largest solution is greater than $(n-3)/2$, and so falls outside of the range for $\tilde{C}_S(p,q)$. 

We know that with probability $\alpha$, a pair $\{p,q\}$ chosen uniformly at random will have $\tilde{C}_S(p,q)\ge(1/2-\sqrt{(1+2\alpha)/12})n + O(1)$. By Lemma~\ref{lem:alg1}, the value of $\tilde{C}_S(p,q)$ can be computed in $O(n\log n)$ time. The expected number of pairs that must be tested before a success is $\alpha^{-1}$, leading to an expected running time of $O(\alpha^{-1}n\log n)$.

If we choose $\log_{(1-\alpha)^{-1}}(n)$ pairs, then the probability of \emph{no} successes is $(1-\alpha)^{\log_{(1-\alpha)^{-1}}(n)}={1}/{n}$. Thus the probability of success is $1-1/n$, and with high probability we will find a pair with $\tilde{C}_S(p,q)\ge (1/2-\sqrt{(1+2\alpha)/12})n + O(1)$ in $O(n\log^2 n)$ time.
\end{proof}

There is a tradeoff in the value of $\alpha$ that we choose. If $\alpha$ is very close to $1$, meaning the probability of choosing a good pair randomly is high, then the bound on the number of points inside and outside the disks goes to $0$. On the other hand, as $\alpha$ gets closer to $0$ the bound on the number of points inside and outside the disks approaches the bound of $n/4.7$ proved by Ramos and Via\~na.

For instance, if $c=1/2$ then the optimal value of $k$ is roughly $ \left(\frac{1}{2}-\frac{1}{\sqrt{6}}\right)n\approx n/10.9$. Therefore, if we choose a pair $\{p,q\}$ uniformly at random, the probability that $C_S(p,q)\ge n/10.9$ is $1/2$.

We note that by computing $C_S(p,q)$ for every pair $\{p,q\}$ of points, we get a simple brute-force algorithm for finding the pair that maximizes $C_S(p,q)$. Since there are $\binom{n}{2}$ pairs of points this takes $O(n^3\log n)$ time. In the next section we will give an algorithm for finding the pair that maximizes $C_S(p,q)$ that runs in $O(n^2\log n)$ time.
\begin{obs}
Given a set $S$ of $n$ points in the plane, a pair of points $p, q \in S$ that achieves the maximum value of $C_S(p,q)$ or $\tilde{C}_S(p,q)$ over all pairs can be found in $O(n^3 \log n)$ expected time.
\end{obs}
Since $\Pi(n)\ge n/4.7$, we are guaranteed to find a pair with $C_S(p,q)\ge n/4.7$, however for a given point set there may be a pair with $C_S(p,q)$ as large as $n/2$. This set is constructed as follows: place two points $ p $ and $ q $ anywhere in the plane, then each group of $ (n-2)/2 $ points is placed on each side of the line $ L(p,q) $, very close to the bisector between $ p $ and $ q $. Then, to prevent degenerate cases, the points are moved slightly in such a way that general position of the points is preserved, see Figure~\ref{fig:example}. Every disk through $ p $, $ q $ and a third point encloses $ (n-2)/2 $ points.

\begin{figure}
   \centering
   \includegraphics{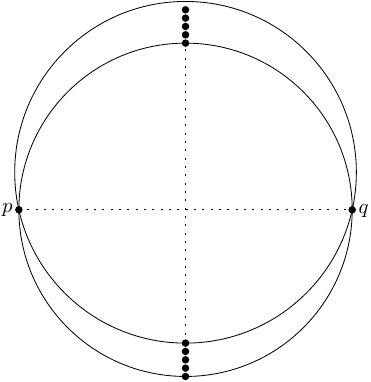}
   \caption{Configuration of $ n $ points where any disk through $ p $ and $ q $ contains at most $ (n-2)/2 $ other points inside it.}
   \label{fig:example}
\end{figure}

\subsection{A quadratic-time algorithm based on higher-order Voronoi diagrams}

In this section we present a quadratic-time algorithm that decides, given a set $S$ and integer $k$, if there exists a pair of points with $\tilde{C}_S(p,q) \ge k$, and finds such a pair if it does exist. Since we know there always exists a pair of points with $\tilde{C}_S(p,q)\ge n/4.7$, we can use this algorithm (with $k$ set to $n/4.7$) to find such a pair in quadratic time. This is a better constant than the algorithm in Theorem~\ref{thm:alg1} for any choice of $\alpha$, but the running time is quadratic instead of $O(n\log n)$. However, we can also use this quadratic algorithm to find the pair that maximizes $\tilde{C}_S(p,q)$ more quickly than the algorithm from the previous section. The new algorithm is based on the observation that the segments of weight $k-1$ are exactly the edges of an order-$k$ Voronoi diagram~\cite{DBLP:journals/tc/Lee82}, see Figure~\ref{fig:voronoi_segments}, for $ k \in \{1, \ldots, n-1\} $. We also need two observations first made by Claverol et al.~\cite{claverol21}. These observations imply that if a bisector contains a segment of weight $k < n/2$, then it must also contain a segment of weight $j$ for all $k \le j \le n/2$.

\begin{obs}
\label{obs:weights}
If one unbounded segment of $b(p,q)$ has weight $j$, then the other unbounded segment has weight $n-j-2$. Furthermore, since the weights of adjacent segments differ by exactly $1$, $b(p,q)$ contains a segment of weight $k$ for all $j\le k\le n-j-2$.
\end{obs}

In order to find a pair with $\tilde{C}_S(p,q)\ge k$, we first compute the order-$(k+1)$ Voronoi diagram and the order-$(n-k-1)$ Voronoi diagram. There must exist some pair of points such that no edge of the order-$(k+1)$ Voronoi diagram is a segment of that bisector. Observation~\ref{obs:weights} implies that this bisector must have minimum weight at least $k+1$.

\begin{figure}
    \centering
    \includegraphics{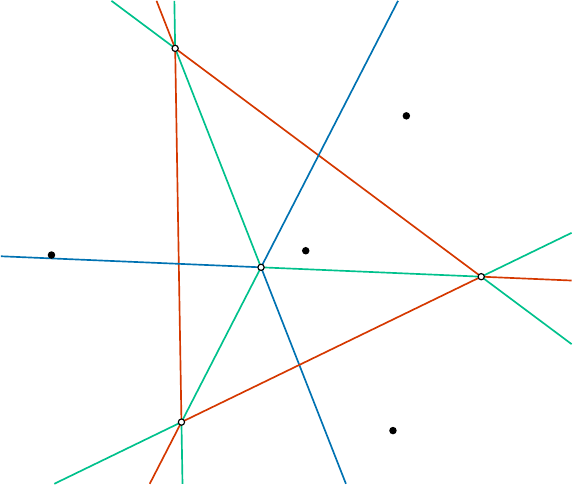}
    \caption{Four points with the order-$1$ Voronoi diagram in red, the order-$2$ Voronoi diagram in green, and the order-$3$ Voronoi diagram in blue. Edges of the order-$k$ Voronoi diagram are the segments of weight $k-1$.}
    \label{fig:voronoi_segments}
\end{figure}

The number of edges of the order-$(k+1)$ Voronoi diagram is known to be $\Theta(nk)$~\cite{DBLP:journals/tc/Lee82}. Since $k$ is roughly $n/4.7$, this is $\Theta(n^2)$. To find a bisector which does not contain an edge of the order-$(k+1)$ Voronoi diagram, we can use a brute-force approach and iterate over the edges of the order-$(k+1)$ Voronoi diagram, marking off which bisector each edge is a segment of. The brute-force algorithm requires quadratic time since there are that many edges, and quadratic space to remember which bisectors have been seen.

The fastest algorithm for constructing the order-$k$ Voronoi diagram in the plane, due to Chan et al.~\cite{DBLP:conf/soda/ChanCZ24}, runs in $O(n\log n+nk)$ expected time. This algorithm can also be used to construct the order-$(n-k-1)$ Voronoi diagram in $O(n\log n+nk)$ expected time.

\begin{algorithm}
\caption{Deciding if there exists a pair $\{p,q\}$ such that $\tilde{C}_S(p,q)\ge k$}
\label{alg:voronoi}
\begin{algorithmic}[1]
\State Compute the order-$(k+1)$ and order-$(n-k-1)$ Voronoi diagrams of $S$
\For{each edge $e$ of $V_{k+1}(S) \cup V_{n-k-1}(S)$}
    \State let $\{p,q\}$ be the pair of points such that $e$ is a segment of $b(p,q)$
    \State mark $\{p,q\}$
\EndFor
\If{there exists an unmarked pair $\{p,q\}$}
    \State \textbf{return} $\{p,q\}$
\Else
    \State \textbf{return} ``no pair''
\EndIf
\end{algorithmic}
\end{algorithm}

\begin{theorem}
Let $S$ be a set of $n$ points in the plane. In $O(n\log n + nk)$ expected time, we can determine if there exists a pair $\{p,q\}$ with $C_S(p,q)\ge k$ or $\tilde{C}_S(p,q)\ge k$, and if such a pair exists we can find it.
\end{theorem}

\begin{proof}
See Algorithm~\ref{alg:voronoi}. Start by computing the order-$(k+1)$ and order-$(n-k-1)$ Voronoi diagrams. If there exists a bisector $b(p,q)$ such that the weight of every segment is greater than $k$ but less than $n-k-2$, then no segment of $b(p,q)$ will appear as an edge of either $V_{k+1}(S)$ or $V_{n-k-1}(S)$. Conversely, by Observation~\ref{obs:weights} if $b(p,q)$ has a segment of weight less than $k$, then it must contain a segment of weight $k$. Likewise, if $b(p,q)$ has a segment of weight greater than $n-k-2$, then it must contain a segment of weight $n-k-2$. Therefore if some bisector does not contain a segment that is either an edge of $V_{k+1}(S)$ or $V_{n-k-1}(S)$, then all segments of that bisector have weights strictly between $k$ and $n-k-2$.

The algorithm iterates over the edges of $V_{k+1}(S)$ and $V_{n-k-1}(S)$, marking which bisector defines each edge. If every bisector is marked after processing every edge, then no pair $\{p,q\}$ such that $\tilde{C}_S(p,q)\ge k$ exists. Otherwise, any unmarked bisector defines a pair such that $\tilde{C}_S(p,q)\ge k$.

Since there are $\binom{n}{2}$ pairs, we require $\Theta(n^2)$ space to keep track of which bisectors are marked. And since the number of edges of $V_{k+1}(S)$ and $V_{n-k-1}(S)$ is known to be $\Theta(nk)$, the running time of iterating over the edges is dominated by the time it takes to construct the diagrams, which is $\Theta(nk + n\log n)$.

If a pair such that $C_S(p,q)\ge k$ is desired instead, we can simply skip the computation of $V_{n-k-1}(S)$ and only consider the edges of $V_{k+1}(S)$ instead.
\end{proof}

Using Algorithm~\ref{alg:voronoi} along with the fact that there always exists a pair $\{p,q\}$ such that $\tilde{C}_S(p,q)\ge n/4.7$, we have the following result.

\begin{theorem}
Let $S$ be a set of $n$ points in the plane. In $O(n^2)$ expected time, a pair of points $p$ and $q$ can be found such that $C_S(p,q)\ge n/4.7$ or $\tilde{C}_S(p,q)\ge n/4.7$.
\end{theorem}

\begin{proof}
Run Algorithm~\ref{alg:voronoi} with $k=n/4.7$.
\end{proof}

For a given set $S$, the value of $\tilde{\Pi}(S)$ lies somewhere between $n/4.7$ and $n/2$. We can therefore perform a binary search on $k$ to determine the value of $\tilde{\Pi}(S)$, and to find a pair that achieves $\tilde{C}_S(p,q)=\tilde{\Pi}(S)$, with at most $O(\log n)$ calls to this algorithm.

\begin{theorem}
Let $S$ be a set of $n$ points in the plane. In $O(n^2\log n)$ expected time, a pair of points $p$ and $q$ that maximizes $\tilde{C}_S(p,q)$ can be found.
\end{theorem}

\subsection{An algorithm for points in convex position}
\label{sec:convex}

In this section we give a linear-time algorithm for finding a pair of points with $C_S(p,q)\ge n/3$ when the set $S$ is in convex position. We assume that the set $S$ is given to us with the points ordered along the convex hull of $S$. Our algorithm is based on the proof of Hayward et al.~\cite{hayward89a} that such a pair of points exists. The following lemma is the basis of the proof.

\begin{lemma}\cite[Lemma~1]{hayward89a}
\label{lem:hayward}
Let $S$ be a set of points, and let $p$ and $q$ be points of $S$. Let $D$ be a disk through $p$ and $q$. The segment $pq$ cuts $D$ into two regions. If both of these regions contain at least $k$ points, then $C_S(p,q) \ge k$.
\end{lemma}

Let $p_1, p_2, \dots, p_n$ be the points of $S$ in order along the convex hull. Notice that for $i<j$, the points $p_i$ and $p_j$ split $S$ into two sets: $\{p_{i+1}, \dots, p_{j-1}\}$ and $\{p_1, \dots, p_{i-1}\}\cup\{p_{j+1},\dots, p_n\}$. The sizes of these two sets are $j-i-1$ and $n-(j-i)-1$, respectively. Thus we can compute, given two points $p_i$ and $p_j$, the sizes of $H(p_i, p_j)$ and $H(p_j, p_i)$ in constant time.

Hayward et al.\ proved that, for any set $S$ of points in convex position, there exists a pair of points $p$ and $q$ such that $C_S(p,q)\ge\ceil{n/3}+1$. These two points must be on the boundary of some enclosing disk~\cite{hayward89a}. Our algorithm simply iterates over all the enclosing disks of $S$ through three points. The following lemma gives a connection between these enclosing disks and the farthest-point Voronoi diagram of $S$.

\begin{lemma}
\label{lem:foo}
Let $S$ be a set of points in the plane. Let $D$ be an enclosing disk through three points. The centre of $D$ is a vertex of the farthest-point Voronoi diagram. Conversely, every vertex of the farthest-point Voronoi diagram is the centre of an enclosing disk through three points.
\end{lemma}

\begin{proof}
Let $D$ be an enclosing disk through three points $p$, $q$, and $r$. By sliding the disk away from $p$ slightly, we get a disk that contains all points of $S$ except for $p$. Thus, the centre of $D$ is adjacent to the farthest-point Voronoi region of $p$. Similarly, the centre of $D$ must be adjacent to the farthest-point Voronoi region of $q$ and $r$, and thus lies on a vertex of the farthest-point Voronoi diagram.

Now consider a vertex $x$ of the farthest-point Voronoi diagram. Since this vertex lies at the intersection of three bisectors, say $b(p,q)$, $b(q,r)$, and $b(r,p)$, it is equidistant from all three points. Consider the disk $D$ with centre $x$ through $p$, $q$, and $r$. If some point of $S$ was not contained in this disk, then $x$ would be in the farthest-point Voronoi region of that point. But $x$ is a vertex of the diagram, so $D$ must be an enclosing disk.
\end{proof}

\begin{theorem}
Given a set $S$ of $n$ points in convex position in the plane, a pair of points $p, q \in S$ such that $C_S(p,q)\ge n/3$ can be found in $O(n)$ time.
\end{theorem}

\begin{proof}
First, we compute the farthest-point Voronoi diagram of $S$. Then we iterate over the vertices of this diagram. Each vertex corresponds to an enclosing disk through three points, say $p_i$, $p_j$, and $p_k$. These three points define three lines, and in constant time we can compute a lower bound on $C_S(p_i,p_j)$, $C_S(p_j,p_k)$, and $C_S(p_k,p_i)$. Lemma~\ref{lem:hayward} tells us that for at least one vertex of the farthest-point Voronoi diagram, we will find a pair of points with $C_S(p,q)\ge\ceil{n/3}+1$.

It takes linear time to compute the farthest-point Voronoi diagram for a set of points in convex position~\cite{DBLP:journals/dcg/AggarwalGSS89}. The farthest-point Voronoi diagram has $O(n)$ vertices so we can also iterate over the vertices in linear time, giving a linear-time algorithm for points in convex position.
\end{proof}

\begin{algorithm}
\label{alg:convex}
\caption{Algorithm for points in convex position}
\begin{algorithmic}[1]
\State Compute the farthest-point Voronoi diagram of $S$
\For{each vertex $v$}
    \State let $p_i$, $p_j$, $p_k$ be the three points defining $v$
    \If{$\min\{\abs{j-i}, n-\abs{j-i}\} - 1 \ge n/3+1$}
        \State \textbf{return} $\{p_i, p_j\}$
    \ElsIf{$\min\{\abs{k-j}, n-\abs{k-j}\} - 1 \ge n/3+1$}
        \State \textbf{return} $\{p_j, p_k\}$
    \ElsIf{$\min\{\abs{i-k}, n-\abs{i-k}\} - 1 \ge n/3+1$}
        \State \textbf{return} $\{p_k, p_i\}$
    \EndIf
\EndFor
\end{algorithmic}
\end{algorithm}

\subsection{An algorithm for diametral disks}

Given two points $p$ and $q$ in the plane, the \emph{diametral disk} $B(p,q)$ is the disk with the segment $pq$ as a diameter. Akiyama et al.~\cite{akiyama96} proved that, in any set $S$ of $n$ points, there always exists a pair $\{p,q\}$ of points of $S$ such that $B(p,q)$ contains at least $n/3$ points of $S$. Their proof leads directly to an algorithm, as they prove that both points of the pair $\{p,q\}$ must lie on the boundary of the smallest enclosing disk of $S$.

\begin{theorem}
Let $S$ be a set of $n$ points in the plane. A pair $\{p,q\}$ of points such that $B(p,q)$ contains at least $n/3$ points of $S$ can be found in $O(n)$ time.
\end{theorem}

\begin{proof}
Begin by computing the smallest enclosing disk $D$ of $S$. This is known to take linear time~\cite{DBLP:journals/siamcomp/Megiddo83a}. If there are only two points $p$, $q$ on the boundary of $D$, then $B(p,q)=D$ and we can return $\{p,q\}$. Otherwise there are three points $p$, $q$, and $r$ on the boundary of $D$. Akiyama et al.~\cite{akiyama96} prove that $D\subset B(p,q)\cup B(q,r)\cup B(r,p)$, so one of the three diametral disks must contain at least $n/3$ points. We can compute the sizes of the sets $B(p,q)\cap D$, $B(q,r)\cap D$, and $B(r,p)\cap D$ in $O(n)$ time.
\end{proof}

\subsection{Algorithms for bichromatic point sets}

Finally we consider the variant of the problem where the set of points is bichromatic. Let $S$ be a set of $n/2$ red and $n/2$ points in the plane. We want to find a pair $\{p,q\}$ such that $p$ is red, $q$ is blue, and $C_S(p,q)\ge cn$ or $\tilde{C}_S(p,q)\ge cn$ for some constant $c>0$. We can use the algorithm of Section~\ref{sec:weightplane} without modification, and we adapt the proof of Theorem~\ref{thm:alg1} to show that we can choose a bichromatic pair uniformly at random and have a constant probability that $\tilde{C}_S(p,q)\ge cn$.

\begin{theorem}
\label{thm:alg1-bichromatic}
Let $0<\alpha<1$ be a constant, and let $S$ be a set of $n/2$ red points and $n/2$ blue points in the plane.
There are at least $\alpha(\frac{n}{2})^2$ pairs $\{p,q\}$ of points of $S$ such that $p$ is red, $q$ is blue, and $\tilde{C}_S(p,q)\ge (1/2-\sqrt{(2+\alpha)/12})n+O(1)$. Such a pair can be found in $O(n \log n)$ expected time, or with probability $1-1/n$ in $O(n\log^2n)$ time.
\end{theorem}

\begin{proof}
Among the $\binom{n}{2}$ bisectors defined by points of $S$, there are $(n/2)^2$ bisectors defined by two points of different colours. Recall that $S_k$ is the number of bisectors that have at least one segment of weight less than $k$ or more than $n-k-2$.
If $S_k < (1-\alpha)(n/2)^2$, then there are at least $\alpha(n/2)^2$ red-blue bisectors with no segments of weight less than $k$ or greater than $n-k-2$.
\end{proof}

We can also use our quadratic-time algorithm to find the red-blue pair $\{p,q\}$ that maximizes either $C_S(p,q)$ or $\tilde{C}_S(p,q)$. The only modification that needs to be made is to only consider bichromatic edges. The running time does not change since the number of bisectors is still $\Theta(n^2)$.

\begin{theorem}
Let $S$ be a set of $n/2$ red points and $n/2$ blue points in the plane. Given an integer $k$, in $O(n\log n+nk)$ expected time we can determine if there exists a red-blue pair $\{p,q\}$ such that $C_S(p,q)\ge k$ or $\tilde{C}_S(p,q)\ge k$, and if such a pair exists we can find it.
\end{theorem}

\begin{proof}
Modify the for loop in Algorithm~\ref{alg:voronoi} to only iterate over edges that are defined by two points of different colour. Also only return an unmarked pair if the two points are of different colour. The full order-$(k+1)$ and order-$(n-k-1)$ Voronoi diagrams are still computed, so the running time does not change.
\end{proof}

\begin{theorem}
Let $S$ be a set of $n/2$ red points and $n/2$ blue points in the plane. In $O(n^2)$ expected time, a pair of points $\{p,q\}$ can be found such that $p$ is red, $q$ is blue, and $C_S(p,q)\ge n/6.8$ (or $\tilde{C}_S(p,q)\ge n/6.8$).
\end{theorem}

\begin{theorem}
Let $S$ be a set of $n/2$ red points and $n/2$ blue points in the plane. In $O(n^2\log n)$ expected time, a red-blue pair $\{p,q\}$ that maximizes either $C_S(p,q)$ or $\tilde{C}_S(p,q)$ can be found.
\end{theorem}

\section{Algorithms in a simple polygon}

In this section, we consider the problem of finding a pair $\{p,q\}$ such that $C_S(p,q)\ge cn$ in the setting of a simple polygon. For this section, let $S$ be a set of $n$ points inside a simple polygon $P$ with $m$ vertices, of which $r$ are reflex and $m-r$ are convex.

Our algorithms compute shortest paths and order-$k$ geodesic Voronoi diagrams inside of a simple polygon. Given a simple polygon $P$ with $m$ vertices and $r$ reflex vertices, Aichholzer et al.~\cite{DBLP:journals/ijcga/AichholzerHKPV14} show that it is possible to compute a simple polygon $P'$ that contains $P$, has the same reflex vertices as $P$, and has a total of $O(r)$ vertices. The problems of computing shortest paths and Voronoi diagrams are equivalent in $P$ and $P'$, in the sense that the shortest path between $p$ and $q$ in $P$ is equal to the shortest path between $p$ and $q$ in $P'$. Because of this, we assume that $P$ is preprocessed using this algorithm, which takes $O(m)$ time. We also assume that the shortest-path data structure of Guibas and Hershberger~\cite{DBLP:journals/jcss/GuibasH89} has been computed for $P'$. This data structure can tell us, in $O(\log r)$ time, the length of the shortest path between any two points in $P'$, and it can additionally report the vertices of the shortest path in time proportional to the number of vertices.

The following lemma due to Pollack, Sharir, and Rote~\cite{DBLP:journals/dcg/PollackSR89} will be useful in our analysis.

\begin{lemma}\cite[Lemma~1]{DBLP:journals/dcg/PollackSR89}
\label{lem:pollack}
Let $ \triangle abc $ be a geodesic triangle in $ P $. As $x$ varies along $g(b,c)$, $ \lvert g(a,x)\rvert $ is a convex function of $\lvert g(b,x)\rvert$, and $ \lvert g(a,x)\rvert < \max\{\lvert g(a,b)\rvert, \lvert g(a,c)\rvert\}$.
\end{lemma}

\subsection{Computing the minimum weight segment on a bisector in a simple polygon}

There are several challenges to adapting the algorithm for computing $C_S(p,q)$ given $p$ and $q$. Even the first step of computing the centre of the disk through three points is nontrivial in a simple polygon. In fact, there may not even be a geodesic disk with three given points on its boundary~\cite{DBLP:journals/dcg/AronovFW93}. However, Oh and Ahn~\cite{DBLP:journals/dcg/OhA20} proved that if there exists a geodesic disk with three given points on its boundary, then it is unique. They also show how to find the centre of this disk in $O(\log^2r)$ time. Their algorithm assumes that the shortest-path data structure of Guibas and Hershberger~\cite{DBLP:journals/jcss/GuibasH89} has been constructed. Constructing this data structure takes $O(m)$ preprocessing time.

One important thing to note is that, since we are working in a polygon $P'$ that contains~$P$, we must ensure that the centre of this disk is actually inside $P$. Using standard planar point location data structures this takes $O(\log m)$ time per point~\cite{DBLP:books/lib/BergCKO08}, with $O(m)$ preprocessing time. This is the only operation we need that does not give the same answer in $P'$ as in $P$.

Once we have found the centres of the geodesic disks through $p$, $q$, and the rest of the points $x_i$ for $1\le i\le n-2$, we need to sort them along the bisector. Again in the plane this is a simple task, but in a simple polygon it is not obvious how to do this efficiently. However, all of these centres lie on $b(p,q)$. The next lemma shows that we can sort the points along $b(p,q)$ by sorting their distances to $p$ (or $q$).

\begin{lemma}
Let $p$ and $q$ be two points in a simple polygon $P$. Let $a$ and $b$ be two points on $b(p,q)$ such that $a$ is between $b$ and the midpoint of $g(p,q)$ (the point where $g(p,q)$ and $b(p,q)$ intersect). Then $ \lvert g(p,a)\rvert < \lvert g(p,b)\rvert$.
\end{lemma}

\begin{proof}
Consider the disk centred at $p$ with radius $ \lvert g(p,b)\rvert$. Our goal is to show that $a$ is inside this disk. We can, without loss of generality, assume that $a$ is inside $\triangle pmb$, where $ m $ is the midpoint of $g(p,q)$. Otherwise, consider $q$ instead of $p$.

We know that $ \lvert g(p,a)\rvert < \max\{\lvert g(p,m)\rvert, \lvert g(p,b)\rvert\}$ by Lemma~\ref{lem:pollack}. We also know that $ \lvert g(p,m)\rvert < \lvert g(p,b)\rvert$. Suppose otherwise, then the concatenation of $g(p,b)$ and $g(b,q)$ would be a shortest path from $p$ to $q$ but shortest paths are unique. Therefore $ \lvert g(p,a)\rvert < \lvert g(p,b) \rvert$.
\end{proof}

The only other piece of information needed is which side of the extension path $\ell(p,q)$ a given centre $o$ is on. The data structure of Guibas and Hershberger can tell us this in $O(\log r)$ time, by examining the paths $g(p,q)$ and $g(p,o)$ and looking at the first vertex where they diverge. So we can split the centres into two sets, those on one side of $g(p,q)$ and those on the other, and then sort each set by their distance to $p$.

\begin{lemma}
\label{lem:alg1geo}
Let $S$ be a set of $n$ points in a simple polygon $P$, and let $p$ and $q$ be points of $S$. The value of $C_S(p,q)$ can be computed in $O(m+n(\log^2r+\log m+\log n))$ time.
\end{lemma}

\begin{proof}
Computing the centre $o_i$ of the geodesic disk through $p$, $q$, and $x_i$, for some $1\le i\le n-2$, takes $O(\log^2r)$ time. We also must check that this centre is in $P$, which takes $O(\log m)$ time. If the centre is in $P$, then we compute which side of $\ell(p,q)$ it is on and $\abs{g(p,o_i)}$ in $O(\log r)$ time. We repeat this for each of the $n-2$ points $x_i$, for $1\le i\le n-2$. Once we have the centres in two sets, we can sort them in $O(n\log n)$ time since we have already computed the distances $\abs{g(p,o_i)}$. Finally, once the points are sorted, we can proceed as in Section~\ref{sec:weightplane}. The information needed to know whether to increment or decrement $\omega(s)$ when moving from one segment to the next has already been computed: if $x_i$ is to the left or right of $g(p,q)$. So the remaining work takes $O(n)$ time. Add to that the $O(m)$ preprocessing time for constructing the data structures, and we get a grand total of $O(m+n(\log^2r+\log m+\log n))$.
\end{proof}

We now show that there are many pairs of points with $C_S(p,q)\ge cn$, for $c\approx 1/11.1$. Unfortunately the proof from the Euclidean setting does not translate to the geodesic setting, because Ramos and Via\~na used properties of the transformation that maps a point $(x,y)$ in $\R^2$ to the point $(x,y,x^2+y^2)$ in $\R^3$ in their proof. There does not seem to be an easy way to generalize this mapping to the geodesic setting. However we can instead look to the proofs of Bose et al.~\cite{wadspaper}. They also bound the number of bisectors that have segments of low weight, but in a different way. We note that, in a simple polygon, the lower bounds on $\Pi(S)$ and $\tilde\Pi(S)$ are different.

Consider the union of all segments of weight $i$, for $0\le i\le k$, for some fixed integer $k\leq n-2$. A given bisector $b(p,q)$ may have many segments of weight at most $k$. Two adjacent segments of weight at most $k$ are part of the same \emph{connected component}. Let $\lambda_k$ be the total number of connected components over all bisectors. Bose et al.~\cite{wadspaper} proved that $\lambda_k\le3kn-\frac{5}{2}k^2-\frac{5}{2}k$. Notice that if $\lambda_k < \binom{n}{2}$, then there must be at least one bisector that has no connected component. In other words, there is some bisector where no segment has weight $k$ or smaller. That implies the existence of a pair with $C_S(p,q)\ge k$.

\begin{theorem}
\label{thm:geodesicrandom}
Let $0<\alpha<1$ be a constant, and let $S$ be a set of $n$ points in a simple polygon $P$ with $m$ vertices, $r$ of which are reflex. With probability $\alpha$, a pair of points $\{p, q\} \in S$ chosen uniformly at random will have $C_S(p,q)\ge (3-\sqrt{4+5\alpha})n/5$.
\end{theorem}

\begin{proof}
The proof is similar to the proof of Theorem~\ref{thm:alg1}. If $\lambda_k<\alpha\binom{n}{2}$ then there are at least $(1-\alpha)\binom{n}{2}$ bisectors that do not contain a segment of weight $k$ or less. Since $\lambda_k\le3kn-\frac{5}{2}k^2-\frac{5}{2}k$, this is true as long as
\begin{equation*}
    k<\frac{3n-\frac{5}{2}-\sqrt{(4+5\alpha)n^{2}-(10+5\alpha)n+\frac{25}{4}}}{5}
    =\left(\frac{3-\sqrt{4+5\alpha}}{5}\right)n+O\left(1\right). \qedhere
\end{equation*}
\end{proof}

As for the variant with at least $k$ points inside \emph{and} outside any disk with $p$ and $q$ on the boundary, we can define $\overline\lambda_{n-k}$ to be the total number of connected components over all bisectors of the union of all segments of weight greater than $n-k-1$. If $b(p,q)$ does not have a connected component, then any disk with $p$ and $q$ on the boundary will contain at most $n-k-1$ points. Bose et al.~\cite{DBLP:journals/corr/abs-2506-06477} proved that $\overline\lambda_{n-k}\le4kn-\frac{7}{2}k^2-\frac{7}{2}k$.

\begin{theorem}
\label{thm:geodesicrandominout}
Let $0<\alpha<1$ be a constant, and let $S$ be a set of $n$ points in a simple polygon $P$ with $m$ vertices, $r$ of which are reflex. With probability $\alpha$, a pair of points $\{p, q\} \in S$ chosen uniformly at random will have $\tilde{C}_S(p,q)\ge (7-\sqrt{37+12\alpha})n/12$.
\end{theorem}

\begin{proof}
In order to ensure the existence of a bisector with no segments of weight at most $k$ or at least $n-k-1$, we must have $\lambda_k + \overline\lambda_{n-k} < \binom{n}{2}$. Combining the upper bounds on $\lambda_k$ and $\overline\lambda_{n-k}$ together, we get $\lambda_k+\overline\lambda_{n-k}<7kn-6k^2-6k$. Therefore we have $\lambda_k + \overline\lambda_{n-k} < \binom{n}{2}$ if
\begin{equation}
    k<\frac{7n-6-\sqrt{(37+12\alpha)n^{2}-(72+12\alpha)n+36}}{12}=\left(\frac{7-\sqrt{37+12\alpha}}{12}\right)n+O(1). \qedhere
\end{equation}
\end{proof}

\begin{theorem}
Let $0<\alpha<1$ be a constant, and let $S$ be a set of $n$ points in a simple polygon $P$ with $m$ vertices, $r$ of which are reflex. A pair of points $\{p, q\} \in S$ such that $C_S(p,q)\ge (3-\sqrt{4+5\alpha})n/5$ or a pair of points such that $\tilde{C}_S(p,q)\ge (7-\sqrt{37+12\alpha})n/12$ can be found in $O(m+n(\log^2r+\log m+\log n))$ expected time, or with probability $1-1/n$ in $O(m+n\log n(\log^2r+\log m+\log n))$ time.
\end{theorem}

\begin{proof}
We know that with probability $\alpha$, a pair chosen at random will have the required property (by either Theorem~\ref{thm:geodesicrandom} or Theorem~\ref{thm:geodesicrandominout}). By Lemma~\ref{lem:alg1geo}, the minimum and maximum number of points in any disk through the pair can be computed in $O(m+n(\log^2r+\log m+\log n))$ time. The expected number of pairs that must be tested before a success is $\alpha^{-1}$, leading to an expected running time of $O(m+n(\log^2r+\log m+\log n))$.

If we choose $\log_{(1-\alpha)^{-1}}(n)$ pairs, then the probability of \emph{no} successes is $(1-\alpha)^{\log_{(1-\alpha)^{-1}}(n)}={1}/{n}$. Thus the probability of success is $1-1/n$. Running the algorithm of Lemma~\ref{lem:alg1geo} $O(\log n)$ times takes $O(m+n\log n(\log^2r+\log m+\log n))$ time.
\end{proof}

\subsection{Algorithms based on higher-order geodesic Voronoi diagrams}

\begin{figure}
    \centering
    \includegraphics{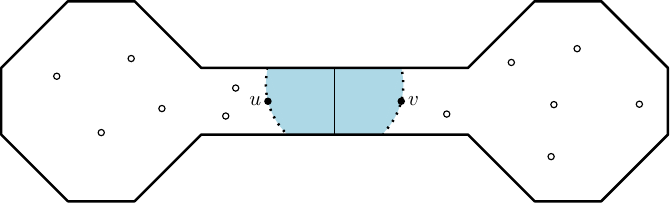}
    \caption{A simple polygon and two points $u$ and $v$ whose bisector consists of a single segment of weight $0$. No geodesic disk with $u$ and $v$ on the boundary contains another point of the set.}
    \label{fig:weight-zero-bisector}
\end{figure}

In the geodesic setting, the observations of Claverol et al.\ do not apply since bisectors are bounded by the boundary of the polygon. If $b(p,q)$ does not have a segment of weight $k$, we cannot necessarily conclude that $b(p,q)$ does not have a segment of weight less than $k$. It is possible that $b(p,q)$ consists of a single segment of weight $0$, see Figure~\ref{fig:weight-zero-bisector} for example. In order to adapt the brute force algorithm for the geodesic setting, we would need to construct the order-$j$ geodesic Voronoi diagrams for all $1\le j\le k$ and take their union.

Oh and Anh~\cite{DBLP:journals/dcg/OhA20} give an algorithm for the order-$k$ geodesic Voronoi diagram in a simple polygon which takes $O(k^2n\log n\log^2r+\min\{rk,r(n-k)\})$ time. Bose et al.\ \cite{wadspaper} proved that $\Pi(n)\ge n/5$ in the geodesic setting, so computing the order-$k$ geodesic Voronoi diagrams from $k=1$ to $k=n/5$ gives an $O(n^4\log n\log^2r + rn^2)$ time algorithm. On the other hand, the brute-force algorithm of trying all $\binom{n}{2}$ pairs and computing $C_S(p,q)$ for each of them takes $O(m + n^3(\log^2r + \log m + \log n))$ time, which is faster by a linear factor. The brute-force algorithm also has the advantage of finding the actual pair that maximizes $C_S(p,q)$, which the order-$k$ Voronoi diagram approach does not necessarily do.

\subsection{An algorithm for points in geodesically convex position}

As in the Euclidean case, when the points of $ S $ are in geodesically convex position, we assume they are given in order along the geodesic convex hull of $S$. In order to generalize the algorithm for points in convex position to the geodesic setting, we need to prove Lemma~\ref{lem:hayward} holds in the geodesic setting.

\begin{lemma}
Let $S$ be a set of points in a simple polygon $P$, and let $p$ and $q$ be points of $S$. Let $D$ be a geodesic disk through $p$ and $q$. The path $g(p,q)$ cuts $D$ into two regions. If both of these regions contain at least $k$ points, then $C_S(p,q) \ge k$.
\end{lemma}

\begin{proof}
Let $c$ be the centre of $D$. Consider another point $c'$ on $b(p,q)$ to the left of $\ell(p,c)$. Let $D'$ be the disk centred at $c'$ with radius $\abs{g(c',p)}$. Then by~\cite[Lemma 1]{wadspaper}, any point in $D$ to the left of $\ell(p,q)$ is contained in $D'$, so $D'$ contains at least $k$ points of $S$. The same argument shows that a disk $D''$ through $p$ and $q$ centred at a point $c''$ on $b(p,q)$ to the right of $\ell(p,c)$ must also contain at least $k$ points. So, any disk through $p$ and $q$ contains at least $k$ points.
\end{proof}




In the plane, the pair of points $\{p,q\}$ such that $C_S(p,q)\ge n/3$ must be two of three points on the boundary of some enclosing disk through three points. This enclosing disk corresponds to a vertex of the farthest-point Voronoi diagram. In a simple polygon, the farthest-point Voronoi diagram may not have any vertex. However, in that case then there always exists a pair of points $\{p,q\}$ such that the \emph{diametral disk} through $p$ and $q$, that is the one centred on the midpoint of $g(p,q)$, contains the entire set $S$~\cite{wadspaper}. Note that such a disk is centred on an edge of the farthest-point geodesic Voronoi diagram. So we compute the farthest-point geodesic Voronoi diagram, and if it has at least one vertex, we can proceed as in Section~\ref{sec:convex}, checking each vertex in constant time. If it does not have any vertex, then we can check each pair of points that corresponds to an edge of the diagram in constant time.

The farthest-point geodesic Voronoi diagram of a set of $n$ points in a simple polygon with~$m$ vertices can be computed in $O(m + n\log n)$ time~\cite{DBLP:journals/dcg/Wang23}. Notice that here, using the polygon simplification of Aichholzer et al.\ would not change the running time, as it takes $O(m)$ time to compute the subsuming polygon $P'$. It is also known that the farthest-point geodesic Voronoi diagram has $O(n)$ vertices and edges. Since the correspondence between vertices of the farthest-point geodesic Voronoi diagram and enclosing disks with three points on their boundary given in Lemma~\ref{lem:foo} holds in a simple polygon, we can compute the farthest-point geodesic Voronoi diagram and use the same algorithm to find a pair of points with $C_S(p,q)\ge n/3$.

\begin{theorem}
Given a set $S$ of $n$ points in geodesically convex position in a simple polygon $P$ with $m$ vertices, a pair of points $p,q\in S $ such that $C_S(p,q)\ge n/3$ can be found in $O(m+n\log n)$ time.
\end{theorem}

\subsection{An algorithm for diametral disks}

Similarly to the Euclidean setting, there always exists a pair $\{p,q\}$ in a given set $S$ of points in a simple polygon such that $B(p,q)$ contains at least $n/3$ points of $S$ (in this section, $B(p,q)$ will refer to the \emph{geodesic} disk with $b(p,q)$ as a diameter). Bose et al.~\cite{wadspaper} prove that, in any set $S$ of $n$ points in a simple polygon, there always exists either: an enclosing geodesic disk with three points $p$, $q$, $r$ on the boundary such that the centre of the disk is contained in $\triangle pqr$, or a pair $p$, $q$ of points such that $B(p,q)$ is an enclosing geodesic disk. In the case where three points are on the boundary of the enclosing geodesic disk, they prove that one of $B(p,q)$, $B(q,r)$, or $B(r,p)$ contains at least $n/3$ points of $S$.

\begin{theorem}
Let $S$ be a set of $n$ poins in a simple polygon $P$ with $m$ vertices. A pair of points $p,q\in S$ such that $B(p,q)$ contains at least $n/3$ points of $S$ can be found in $O(m + n\log n)$ time.
\end{theorem}

\begin{proof}
Begin by constructing the farthest-point geodesic Voronoi diagram of $S$, and the shortest path data structure. Each vertex $v$ of the farthest-point diagram corresponds to an enclosing geodesic disk with three points $p$, $q$ and $r$ on its boundary. In $O(\log n)$ time we can test whether $v$ is inside $\triangle pqr$ using the shortest path data structure. We can test all vertices of the diagram in $O(n \log n)$ time. If we find a vertex $v$ such that $v$ is inside $\triangle pqr$, then one of the disks $B(p,q)$, $B(q,r)$, or $B(r,p)$ will contain at least $n/3$ points. We can find the midpoints of these disks in $O(\log n + m)$ time by constructing the shortest paths explicitly, then counting the number of points inside a given disk takes $O(n \log n)$ time.

If no vertex is inside the triangle defined by the three corresponding points, then there must be a diametral enclosing disk. We can find this by looking at the edges of the diagram. Let $e$ be an edge of the diagram defined by the two points $p$ and $q$. Note that $B(p,q)$ is an enclosing disk if and only if $b(p,q)$ intersects $e$, which we can test in $O(\log n)$ time by checking if the endpoints of $e$ are on opposite sides of $b(p,q)$. Since there are $O(n)$ edges in the diagram, this takes $O(n\log n)$ time in total.
\end{proof}

\subsection{Algorithms for bichromatic point sets}

We now consider the variant of the problem where the set of points is bichromatic. Let $S$ be a set of $n/2$ red and $n/2$ points in a simple polygon. We want to find a pair $\{p,q\}$ such that $p$ is red, $q$ is blue, and $C_S(p,q)\ge cn$ for some constant $c>0$. 

\begin{theorem}
\label{thm:alg1-bichromatic-polygon}
Let $0<\alpha<1$ be a constant, and let $S$ be a set of $n/2$ red points and $n/2$ blue points in a simple polygon $P$ with $m$ vertices, $r$ of which are reflex.
There are at least $\alpha(\frac{n}{2})^2$ pairs $\{p,q\}$ of points of $S$ such that $p$ is red, $q$ is blue, and $C_S(p,q)\ge (3/5 - \sqrt{(13+5\alpha)/50})n+O(1)$. Such a pair can be found in $O(m + n(\log^2r+\log m+\log n))$ expected time, or with probability $1-1/n$ in $O(m + n\log n(\log^2r+\log m+\log n))$ time.
\end{theorem}

\begin{proof}
Recall that $\lambda_k$ is the number of connected components over all bisectors in the union of all segments of weight $i$, for $0\le i\le k$, and that $\lambda_k \le 3kn - \frac52k^2 - \frac52k$. If $\lambda_k<(1-\alpha)(\frac{n}{2})^2$, then there must be at least $\alpha(\frac{n}{2})^2$ red-blue bisectors that do not have any connected components, which implies every segment on each of those bisectors has weight greater than $k$. The largest value of $k$ for which this is true is
\begin{equation*}
    k < \frac{3n-\frac{5}{2}-\sqrt{(\frac{13}{2}+\frac{5}{2}\alpha)n^2-15n+\frac{25}{4}}}{5}=\bigg(\frac{3}{5}-\sqrt{\frac{13+5\alpha}{50}}\bigg)n + O(1). \qedhere
\end{equation*}
\end{proof}

Similarly to the monochromatic case, we have a separate theorem for the in-out version of the problem for the geodesic setting.

\begin{theorem}
\label{thm:alg1-bichromatic-inout-polygon}
Let $0<\alpha<1$ be a constant, and let $S$ be a set of $n/2$ red points and $n/2$ blue points in a simple polygon $P$ with $m$ vertices, $r$ of which are reflex.
There are at least $\alpha(\frac{n}{2})^2$ pairs $\{p,q\}$ of points of $S$ such that $p$ is red, $q$ is blue, and $\tilde{C}_S(p,q)\ge ((7-\sqrt{43+6\alpha})/12)n + O(1)$. Such a pair can be found in $O(m + n(\log^2r+\log m+\log n))$ expected time, or with probability $1-1/n$ in $O(m + n\log n(\log^2r+\log m+\log n))$ time.
\end{theorem}

\begin{proof}
We need $\lambda_k+\overline\lambda_{n-k}<(1-\alpha)(\frac{n}{2})^2$. We can use the upper bound of $\lambda_k + \overline\lambda_{n-k} \le 7kn - 6k^2 - 6k$. Combining these and solving for $k$ gives
\begin{equation*}
    k<\frac{7n-6}{12}-\sqrt{\frac{(43+6\alpha)n^2-84n+36}{12}} = \bigg(\frac{7-\sqrt{43+6\alpha}}{12}\bigg)n + O(1). \qedhere
\end{equation*}
\end{proof}

\section{Conclusion}

We presented several different algorithms for finding a pair $p$ and $q$ such that $C_S(p,q)\ge cn$, for some $c>0$, given a set $S$ of $n$ points in the plane. The first two of these algorithms are for sets of points in general position. The other algorithms are for variations of the problem: when the points are in convex position, when we are only interested in finding a pair such that the diametral disk of that pair contains many points, and when the set of points is bichromatic. We also studied how these algorithms can be generalized to the geodesic setting, where $S$ is a set of $n$ points inside a simple polygon $P$.

Our algorithms for points in convex position and for the diametral version of the problem match the known value of $\overline\Pi(n)$, so they can only be improved by finding a faster algorithm. Such an algorithm cannot use the farthest-point geodesic Voronoi diagram, since optimal algorithms for constructing these diagrams match our current running times. The situation is more open for points in general position. Is there a sub-quadratic algorithm that is guaranteed to find a pair with $C_S(p,q)\ge\Pi(n)$? Another interesting problem is improving the algorithm that finds the pair that maximizes $C_S(p,q)$ for a given point set~$S$.

\bibliography{ref}

\end{document}